\documentclass[letter,11pt]{article}

\usepackage[english]{babel}
\usepackage[utf8x]{inputenc}
\usepackage[T1]{fontenc}
\usepackage{amsfonts}

\usepackage[margin = 1in]{geometry}

\usepackage{amsmath}
\usepackage{amsthm}
\usepackage{algpseudocode,algorithm,algorithmicx}
\usepackage{bbm}
\usepackage{graphicx}
\usepackage[colorinlistoftodos]{todonotes}
\usepackage[colorlinks=true, allcolors=blue]{hyperref}
\usepackage{thmtools}
\usepackage{thm-restate}

\newtheorem{theorem}{Theorem}
\newtheorem{lemma}[theorem]{Lemma}
\newtheorem{fact}[theorem]{Fact}
\newtheorem{definition}[theorem]{Definition}
\newtheorem{corollary}[theorem]{Corollary}

\newcommand{\cond}{\mathcal{E}}
\newcommand{\ex}{\mathbb{E}}

\newcommand{\f}{f_{A^*}}
\newcommand{\samp}[1]{(s_1, s_2, \mathcal{E}) \sim ID(#1)}
\newcommand{\LR}{\mathcal{LBR}}
\newcommand{\SR}{\mathcal{SBR}}

\title{Near-Linear Time Edit Distance for Indel Channels}

\author{
  Arun Ganesh
  \footnote{Department of Electrical Engineering and Computer Sciences, UC Berkeley. Email: \texttt{arunganesh@berkeley.edu}. Supported by NSF Award CCF-1535989.}
  \and
  Aaron Sy
  \footnote{Department of Electrical Engineering and Computer Sciences, UC Berkeley. Email: \texttt{raaronsy@gmail.com}. Supported by NSF Award CCF-1535989.}
  }
\date{} 
\begin{document}

\maketitle

\begin{abstract}
    We consider the following model for sampling pairs of strings: $s_1$ is a uniformly random bitstring of length $n$, and $s_2$ is the bitstring arrived at by applying substitutions, insertions, and deletions to each bit of $s_1$ with some probability. We show that the edit distance between $s_1$ and $s_2$ can be computed in $O(n \ln n)$ time with high probability, as long as each bit of $s_1$ has a mutation applied to it with probability at most a small constant. The algorithm is simple and only uses the textbook dynamic programming algorithm as a primitive, first computing an approximate alignment between the two strings, and then running the dynamic programming algorithm restricted to entries close to the approximate alignment. The analysis of our algorithm provides theoretical justification for alignment heuristics used in practice such as BLAST, FASTA, and MAFFT, which also start by computing approximate alignments quickly and then find the best alignment near the approximate alignment. Our main technical contribution is a partitioning of alignments such that the number of the subsets in the partition is not too large and every alignment in one subset is worse than an alignment considered by our algorithm with high probability. Similar techniques may be of interest in the average-case analysis of other problems commonly solved via dynamic programming. 
\end{abstract}
\section{Introduction}
Edit distance is an important string similarity measure whose computation has applications in many fields including computational biology. Its simplest variant is the Levensthein distance, which is the minimum number of insertions, deletions, or substitutions required to turn the first string into the second. A textbook dynamic programming algorithm computes the edit distance between two length $n$ strings in $O(n^2)$ time (see e.g. Section 6.3 of \cite{DasguptaPV08}), and the best known worst-case exact algorithm runs in $O(\frac{n^2}{\ln  ^2 n})$ time \cite{MasekP80}. Assuming the Strong Exponential Time Hypothesis, Backurs and Indyk showed that no $O(n^{2-\epsilon})$ time algorithm exists \cite{BackursI14} for any $\epsilon > 0$, and Bringmann and Künnemann extended this result to the special case of bitstrings, suggesting that these algorithms are near-optimal \cite{BringmannK15}. 

In many practical applications, a quadratic runtime is prohibitively expensive. For example, it was once estimated that using the textbook algorithm to align the full genomes of a human and a mouse (although not a very practical problem) would take 95 CPU years \cite{Frith08}. When the edit distance is small, one can do better. An immediate result is that if the edit distance between two length $n$ strings is at most $d$, it can be computed in time $O(nd)$ (by considering only entries in the dynamic programming table which are distance at most $d$ from entries indexed $(i, i)$ for some $i$), and Landau et al. give a more nuanced  algorithm which finds the edit distance in time $O(n + d^2)$ \cite{LandauMS98}. However, when e.g. aligning the sequences of two different species the edit distance can still be as large as $\Omega(n)$, so these results do not offer substantial improvements over the textbook algorithm.

Motivated by this and the aforementioned lower bounds, there have been many efforts to design faster algorithms. Many worst-case approximation algorithms exist for the problem (e.g. \cite{BatuES06, AndoniKO10, AndoniO11, CDGK18}). However, most results give super-constant approximation ratios, and even the known constant approximation ratios are perhaps too large for practical applications. For example, popular knowledge suggests that a 3-approximation algorithm\footnote{The approximation ratio proven by \cite{CDGK18} is 1680, though they conjecture their algorithm is actually a $(3+\epsilon)$-approximation.} for edit distance when applied to genome sequences is not guaranteed to determine that humans are more closely related to dogs than chickens.

However, there is good reason to believe that in biological applications, the subquadratic lower bound is not applicable. Roughly speaking, the lower bounds of \cite{BackursI14, BringmannK15} say that every part of one string must be compared to every part of another string in order to compute the edit distance exactly. In practice, this should rarely be true. e.g. when aligning two genomes, there is good reason to believe that the beginning of the first genome only needs to be compared to the beginning of the second genome. Observations like this motivate the need for \textit{average-case analysis} of edit distance algorithms. There are already several results on average-case analyses of edit distance. For example, \cite{AndoniK08} gives an approximation algorithm when the inputs are chosen adversarially but then perturbed, \cite{Gawrychowski12} gives an exact algorithm when the inputs are compressible, and \cite{Kuszmaul18} gives an approximation algorithm when one of the input strings satisfies a pseudo-randomness condition. Note that all these results require losing an approximation factor (which as mentioned before is undesirable) and/or for fairly specific conditions (such as compressibility) to hold for the input. 

\subsection{Our Contribution}

In this paper, we consider a model for average-case analysis of edit distance called the \textit{indel channel} which is motivated by biological applications. In this model, we generate a random bitstring of length $n$ as our first string (using bitstrings simplifies the presentation, and the results generalize easily to larger alphabets), and then at each position in the string randomly apply each of the three types of mutations (insertion, deletion, substitution) independently with some probability to get the second string. We let $ID(n)$ denote the distribution of pairs of strings and sets of mutations generated by this model. This model of random string mutation is popular as an extension of the CFN model for biological mutations in computational biology, and problems based on the indel channel have been defined and studied in the areas of sequence alignment \cite{Frith19}, phylogenetic reconstruction \cite{DaskalakisR10, AndoniBH10, AndoniDHR12, GaneshZ18}, and trace reconstruction \cite{HolensteinMPW08, NazarovP17, HoldenPP18}. We show that for pairs of strings generated by this model, we can compute their exact edit distance in near-linear time with high probability:

\begin{theorem}[Informal]\label{thm:main-informal}
Let $s_1$ be a uniformly random bitstring of length $n$ and $s_2$ be the bitstring generated by applying substitution, insertion, and deletion to each bit of $s_1$ each uniformly at random and with probability at most some constant. Then with high probability we can compute the edit distance between $s_1$ and $s_2$ in $O(n \ln  n)$ time.
\end{theorem}

\textbf{Our Techniques.} Our algorithm is simple, using only the dynamic programming algorithm as a primitive. The high-level approach is as follows: While we cannot use the dynamic programming algorithm to compute the edit distance between the two strings and get a near-linear time algorithm, we can repeatedly use it to compute the edit distance between two substrings of length $k \ln  n$, where $k$ is a (sufficiently large) constant. Under the indel channel, a substring of length $k \ln  n$ of the first string $s_1$ and the corresponding substring of the second string $s_2$ have low edit distance compared to two random substrings with high probability. So by computing the edit distance between two substrings of length $k \ln  n$, we can determine if the correct alignment places these two substrings close to each other.

We can now use this as a primitive to find an alignment of the two strings that is an approximation of the ``canonical'' alignment, i.e. the alignment corresponding to the insertions and deletions caused by indel channel. If we know bit $i$ of $s_1$ is aligned with bit $j$ of $s_2$, then with high probability there are only $O(\ln  n)$ indices in $s_2$ that bit $i+k \ln  n$ of $s_1$ can be aligned with. Even if we only have an estimate for where bit $i$ of $s_1$ is aligned with in $s_2$ that is $O(\ln n)$ bits off, with high probability the number of indices bit $i+k\ln n$ of $s_1$ can be aligned with is still $O(\ln n)$. So, once we have computed an approximate alignment for the first $i$ bits of $s_1$, we can iteratively extend the approximate alignment by using a small number of edit distance computations on bitstrings of length $O(\ln n)$ to determine approximately where bit $i+k\ln n$ of $s_1$ should be aligned. We note that some past works studying the indel channel in phylogenetic reconstruction use the trivial ``diagonal'' alignment (e.g. \cite{DaskalakisR10, GaneshZ18}) as an approximate alignment. 

Once we have an approximate alignment, our algorithm is straightforward: Use the dynamic programming algorithm, but only compute entries in the dynamic programming table which are close to the approximate alignment. We show that with high probability, the best alignment is close to the canonical alignment suggested by the indel channel, which is close to our approximate alignment, giving the correctness of this algorithm. To show this statement holds, we would like to use the fact that that any alignment which differs significantly from the canonical alignment is better than the canonical alignment with probability decaying exponentially in the difference between the two alignments. However, there are too many alignments for us to conclude by combining this fact with a union bound. Instead, we construct a partition $\mathcal{B}$ of the alignments such that for each element $B$ of the partition $\mathcal{B}$, the alignments in $B$ are structurally similar. Roughly speaking, this lets us argue for each $B$ that with probability much smaller than $1/|\mathcal{B}|$ all alignments in $B$ are not optimal. We can then take a union bound over the subsets in $\mathcal{B}$ instead of over all alignments to get the desired statement.

We note that techniques similar to finding an approximate alignment and then computing the DP table restricted to entries near this alignment are used in heuristics in practice such as BLAST \cite{AltschulGMML90}, FASTA \cite{PearsonL88}, and MAFFT \cite{KatohMKM02}. Our analysis thus can be viewed as theoretical support for these kinds of heuristics.

The rest of the paper is as follows: In Section~\ref{sec:prelim}, we define the indel channel model formally, give some simple probability facts that are useful, and define some terms that appear frequently in the analysis. In Section~\ref{sec:subonly}, as a warm-up we show that in the substitution-only case, the optimal alignment is close to the diagonal. In Section~\ref{sec:alignment} we describe and analyze our algorithm for finding an approximate alignment. In Section~\ref{sec:indels}, we extend the analysis from Section~\ref{sec:subonly} to the general case, completing the proof of Theorem~\ref{thm:main-informal}.
\section{Preliminaries and Definitions}\label{sec:prelim}

To simplify the presentation, we will often treat possibly non-integer numbers like $\ln n, k \ln n$ and $n / k \ln n$ as integers without explicitly rounding them first. The correctness of all proofs in the paper is unaffected by replacing these quantities by their rounded versions (e.g. $\lceil \ln n \rceil$) where appropriate.

\subsection{Problem Setup}

In this section, we describe the model used to generate the pairs of correlated strings and formally state our main result. We start by sampling a uniformly random bitstring $s_1 \sim \{0, 1\}^n$. We pass $s_1$ through an indel channel to arrive at a new bitstring $s_2$. When passed through the indel channel, for the $j$th bit of $s_1$, $b_j := (s_1)_j$:

\begin{itemize}
    \item $b_j$ is substituted, i.e. flips, with probability $p_s$.
    \item $b_j$ is deleted, with probability \begin{itemize}
        \item $p_d$ if the previous bit $b_{j-1}$ was not deleted,
        \item $q_d > p_d$ if the previous bit $b_{j-1}$ was deleted. (This is similar but not equivalent to deleting a geometric number of bits whenever a deletion occurs)
    \end{itemize}
    That is, whenever a bit $b_j$ is deleted, an additional number of bits equal to roughly a geometric random variable with mean $1/(1 - q_d)$ are deleted to the right of $b_j$.
    \item An insertion event occurs with probability $p_i$, inserting a uniformly random bit string $t \sim \{0, 1\}^I$ with length $I\sim \operatorname{Geo} \left(1 - {q_i}\right)$ ($I$ has mean $1/(1-q_i)$) to the right of $b_j$. Inserted bits are not further acted upon by the indel channel.
\end{itemize}

We call each of these edits, and use $\mathcal{E}$ to denote the set of edits occurring in the indel channel. Each mutation happens independently for each bit and across different bits. As mentioned before, this definition of the indel channel is chosen to parallel models in both the computer science theory and computational biology communities that account for splicing in/out entire subsequences rather than individual sites (e.g. see \cite{Frith19} for an example of a model for mutation which uses geometric indel lengths; of course, setting $q_d = p_d, q_i = 0$ gives the setting where only single bits are spliced in/out). We require that

\begin{equation} \label{eq:upperbounds}
p_s \leq \rho_s, \qquad\qquad\qquad p_d \leq \rho_d, \frac{1-\rho_d}{1-q_d} \leq \rho_d', \qquad\qquad\qquad p_i \leq \rho_i, \frac{1}{1-q_i} \leq \rho_i'. 
\end{equation}

for some small constants $\{\rho\}:= \{\rho_s, \rho_d, \rho_d', \rho_i, \rho_i'\}$. Our lemma/theorem statements will implicity assume \eqref{eq:upperbounds} holds, and our proofs will specify certain inequalities which must hold for the values $\{\rho\}$, thus specifying a range of values for the mutation probabilities for which our algorithm is proven to work. We do not attempt to the optimize the values of $\{\rho\}$ for which our algorithm works, but will state the exact inequalities that need to hold for $\{\rho\}$ when it is convenient to do so. We use $ID(n)$ to denote the distribution of tuples $(s_1, s_2, \mathcal{E})$ arrived at by this process for some $p, q$ values - we often make statements about $\samp{n}$ which apply for any realization of the $p, q$ values satisfying the constraints given by $\{\rho\}$, in which case we will not specify what these values are. Given $\samp{n}$ we want to compute the edit distance $ED(s_1, s_2)$ between $s_1$ and $s_2$ as quickly as possible. For simplicity, we specifically use the Levenshtein distance in our proofs, but they can easily be generalized to other sets of penalties for edits. We now formally restate Theorem~\ref{thm:main-informal} as our main result:

\begin{theorem}\label{thm:main}
Assuming \eqref{eq:upperbounds} holds for certain constants $\{\rho\}$, there exists a (deterministic) algorithm running in time $O(n \ln n)$ that computes $ED(s_1, s_2)$ for $\samp{n}$ with probability $1-n^{-\Omega(1)}$.
\end{theorem}

\subsection{Probability Facts}
We start with some basic probability facts appearing in our analysis. We denote the number of ways to sort $a+b+c$ elements into three groups of size $a, b, c$, i.e. trinomial, by $\binom{a+b+c}{a, b, c}$. This of course equals $\frac{(a+b+c)!}{a!b!c!}$. When the trinomial appears, we use Stirling's approximation to bound its value:

\begin{fact}[Stirling's approximation]
$\sqrt{2 \pi} n^{n+1/2} e^{-n} \leq n! \leq e n^{n+1/2} e^{-n}.$
\end{fact}

We do not aim to optimize constants, so we will use the following standard simplified Chernoff bound in our proofs:

\begin{fact}[Chernoff bound]\label{fact:chernoff}
Let $X_1 \ldots X_n$ be independent Bernoulli random variables and $X = \sum_{i=1}^n X_i$ and $\mu = \ex[X]$. Then for $0 < \epsilon < 1$:

$$Pr[X \geq (1+\epsilon)\mu]\leq e^{-\frac{\epsilon^2 \mu}{3}}, \qquad Pr[X \leq (1-\epsilon)\mu] \leq e^{-\frac{\epsilon^2 \mu}{2}}.$$
\end{fact}

We will also use the following simplified negative binomial tail bound:

\begin{fact}[Negative binomial tail bound]
Let $X \sim NBinom(n,p)$, i.e. $X$ is a random variable equal to the number of probability $p$ success events needed before $n$ successes are seen. Then for $k > 1$:
$$Pr[X \geq kn/p] \leq e^{-\frac{kn(1-1/k)^2}{2}}. $$
\end{fact}
\begin{proof}
This follows from noticing that $Pr[X \geq kn/p] = Pr[Binom(kn/p, p) < n]$ and applying a Chernoff bound with $\epsilon = 1 - 1/k$.
\end{proof}

We'll chain together these facts to get a tail bound for a binomial number of geometric random variables:

\begin{lemma}\label{lemma:binomofgeo}
Consider $X \sim NBinom(m, q)$ where $m = \sum_{i=1}^t m_i$, $m_i \sim Bern(p_i)$, i.e. $X$ is a random variable obtained by first sampling $m$, the sum of $t$ independent Bernoullis, and then sampling $X \sim NBinom(m, q)$, where $Bern$ and $NBinom$ denote the standard Bernoulli and negative binomial distributions. Then for $1 < k \leq 4$, $\mu = \sum_{i=1}^t p_i$ :

$$Pr\left[X \geq k \cdot \frac{\mu}{q}\right] \leq e^{-\frac{(\sqrt{k}-1)^2 \mu}{3}} + e^{-\frac{k\mu(1-1/\sqrt{k})^2}{2}}.$$

\end{lemma}
\begin{proof}
We consider two cases for the realization of $m$ and apply tail bounds to each case:

$$Pr\left[X \geq k \cdot \frac{\mu}{q}\right] = Pr\left[X \geq k \cdot \frac{\mu}{q} \land m \geq \sqrt{k}\mu\right] + Pr\left[X \geq k \cdot \frac{\mu}{q} \land m < \sqrt{k}\mu\right]$$
$$\leq Pr\left[m \geq \sqrt{k}\mu\right] + Pr\left[X \geq k \cdot \frac{\mu}{q} | m \leq \sqrt{k}\mu\right].$$

A Chernoff bound gives $Pr[m \geq \sqrt{k}\mu] \leq e^{-\frac{(\sqrt{k}-1)^2 \mu}{3}}$, the negative binomial tail bound  (and noticing that $Pr[X \geq k \cdot \frac{\mu}{q} | m \leq \sqrt{k}\mu]$ is maximized when $m = \sqrt{k}\mu$) gives $Pr[X \geq k \cdot \frac{\mu}{q} | m < \sqrt{k}\mu] \leq e^{-\frac{\sqrt{k}\mu(1-1/\sqrt{k})^2}{2}}$, giving the lemma.
\end{proof}

\subsection{Definitions}
In this section we give definitions that simplify the presentation. There are many identical definitions for solutions to the edit distance problem - we will define solutions as paths through the dependency graph as doing so simplifies the presentation of the analysis.

\begin{definition}\label{def:ed}
Consider the \textbf{dependency graph} of the edit distance dynamic programming table: For two strings $s_1, s_2$ of length $n_1, n_2$, the dependency graph of $s_1, s_2$ has vertices $(i, j)$ for $i \in \{0, 1, \ldots n_1\}, j \in \{0, 1, \ldots n_2\}$ and directed edges from $(i, j)$ to $(i+1, j), (i, j+1)$ and $(i+1, j+1)$ for $i \in [n_1], j \in [n_2]$ if these vertices exist. The edges $\{(i-1, j-1), (i, j)\}$ where $(s_1)_i = (s_2)_j$ have weight 0 and all other edges have weight 1. The \textbf{edit distance} between $s_1$ and $s_2$, denoted $ED(s_1, s_2)$, is the (weighted) shortest path from $(0, 0)$ to $(n_1, n_2)$. 
\end{definition}

For completeness we recall the standard dynamic programming algorithm for edit distance and its restriction to a subset of indices.
\begin{fact}[Textbook Algorithm]
The edit distance between $s_1, s_2$ of length $n_1, n_2$ can be computed in $O(n_1 n_2)$ time by using e.g. the $O(|V| + |E|)$-time\footnote{Note that the dependency graph has $|E| = O(|V|)$.} dynamic programming algorithm for shortest paths in a DAG. In addition, if we know the shortest path in the dependency graph is contained in vertex set $V'$, we can compute the edit distance in $O(|V'|)$ time by applying the dynamic program ``restricted to $V'$'' i.e. by applying it to the dependency graph after deleting all vertices not in $V'$.
\end{fact}

\begin{definition}
An \textbf{alignment} (of two strings $s_1, s_2$) is any path $A = \{(i_1 = 0, j_1 = 0), (i_2, j_2) \ldots (i_{L-1}, j_{L-1}), (i_L= n_1, j_L = n_2)\}$ from $(0, 0)$ to $(n_1, n_2)$ in the dependency graph of $s_1, s_2$. Denote the set of all alignments by $\mathcal{A}$. 
\end{definition}

For convenience, we will abuse notation and sometimes use $A$ to also denote the cost of alignment $A$, e.g. using $A \geq A'$ to denote that the cost of $A$ is at least the cost of $A'$. 

\begin{definition}
For $\samp{n}$, the \textbf{canonical alignment} of $s_1, s_2$, denoted $A^*$, is informally the alignment corresponding to $\mathcal{E}$. More formally, $A^*$ starts at $(0, 0)$, and for each row $i$ of the dependency graph, if the first vertex in $A^*$ in this row is $(i,j)$, we extend $A^*$ as follows according to $\cond$:
\begin{itemize}
    \item If no insertion or deletion occurs on the $i$th bit, we include the edge $\{(i, j), (i+1, j+1)\}$.
    \item If an insertion of $I$ bits occurs on the $i$th bit and no deletion occurs, we include the path $\{(i, j), (i, j+1), \ldots (i,j+I), (i+1, j+I+1)\}$
    \item If a deletion and no insertion occurred on the $i$th bit, we include the edge $\{(i, j), (i+1, j)\}$. \item If an insertion of $I$ bits occurred and a deletion, we include the path $\{(i, j), (i, j+1), \ldots (i,j+I), (i+1, j+I)\}$.
\end{itemize}
\end{definition}

The definition of (canonical) alignments depends on the pair of strings $s_1, s_2$, but throughout the paper usually it will be clear that the pair of strings being referred to is sampled from $ID(n)$, so for brevity's sake we may refer to a canonical alignment without referring to strings, letting the strings be implicit.

Note that the canonical alignment is not necessarily the optimal alignment (in fact, even in the substitution-only case, the substitutions cause the optimal alignment to be one including insertions and deletions with high probability). However, alignments which differ sufficiently from the canonical alignment should not perform better than the canonical alignment with high probability. For alignments which aren't the canonical alignment, we characterize their differences from the canonical alignment in terms of where they break from the canoncial alignment.

\begin{definition}
Fix a canonical alignment $A^*$, and let $A$ be any alignment. A \textbf{break} of $A$ (from $A^*$) is any subpath $(\{(i_1, j_1), (i_2, j_2) \ldots (i_L, j_L)\})$ of $A$ such that $(i_1, j_1)$ and $(i_L, j_L)$ are in $A^*$ but none of $(i_2, j_2)$ to $(i_{L-1}, j_{L-1})$ are in $A^*$. The \textbf{length} of the break is the value $i_L - i_1$.

For $\samp{n}$, a break of alignment $A$ from $(i_1, j_1)$ to $(i_L, j_L)$ is \textbf{long} if its length is at least $k \ln n$ (for a constant $k$ to be specified later) and \textbf{short} otherwise\footnote{Note that in the definition of length, we use $i_L - i_1$ and ignore $j_1, j_L$. This is because with high probability, for all $i, i'$ such that $i' > i + k \ln n$, if the canonical alignment goes through $(i, j)$ and $(i', j')$, $j' - j$ will be within a constant factor of $i' - i$. So defining length as $i_L - i_1$ instead of $j_L - j_1$ will not substantially affect our categorization of which breaks are short or long.}. An alignment is \textbf{good} if it has no long breaks and \textbf{bad} if it has at least one long break.
\end{definition}
Intuitively, short breaks are smaller and might make an alignment better than the canonical alignment, so we can't rule out alignments containing only short breaks in our analysis. On the other hand, long breaks are sufficiently large such that replacing them with the corresponding part of the canonical alignment should be an improvement with high probability. Lastly, we define two functions that take alignments and make them look more like the canonical alignment $A^*$.

\begin{definition}[Short and Long Break Replacement]
We define $\SR:\mathcal{A}\mapsto \mathcal{A}$ as a function from alignments to alignments, such that for any alignment $A$, $\SR(A)$ is the alignment arrived at by applying the following modification to all short breaks in $A$: For a short break from $(i_1, j_1)$ to $(i_L, j_L)$, replace it with the subpath of $A^*$ from $(i_1, j_1)$ to $(i_L, j_L)$. We define $\LR$ analogously, except $\LR$ applies the modification to all long breaks instead of short breaks.
\end{definition}

Note that all alignments in the range of $\LR$ are good by definition. The idea behind these functions and the definitions of good and bad alignments is to use them in the analysis as follows: It is possible to compute the best of the good alignments quickly by only considering a narrow region within the DP table. So it suffices to show any bad alignment is not the best alignment. For a single bad alignment $A$, it is fairly straightforward to show that $A^*$ is better than $A$ with high probability. 
However, there are many bad alignments and thus a simple union bound does not suffice to complete the analysis. We instead use $\LR$ to show that it suffices if all alignments in the range of $\SR$ are not better than $A^*$ with high probability. There are considerably fewer of these alignments and they can be partitioned in a way that is easy to analyze, and so simple counting and probability techniques let us show this holds.
\section{Substitution-Only Case}\label{sec:subonly}

As a warmup, let's consider the easier case when only substitutions are present in the indel channel. In this case, $A^*$ is just the diagonal $\{(0, 0), (1, 1) \ldots (n, n)\}$. We show the following theorem:

\begin{theorem}\label{thm:subonly}
For $\samp{n}$ with $p_i, p_d = 0$, as long as $p_s \leq \rho_s$ where $\rho_s = .028$, there is an $O(n \ln n)$ time algorithm for calculating $ED(s_1, s_2)$ which is correct with probability $1 - n^{-\Omega(1)}$.
\end{theorem}

The algorithm is simple - compute entries of the canonical DP table indexed by $(i, j)$ where $|i - j| \leq k \ln n$, ignoring dependencies on entries for which $|i - j| > k \ln n$. The value of $k$ used in the algorithm and the definition of long breaks will be specified by the analysis, which will determine a lower bound for $k$ needed to make the failure probability sufficiently small. 

We start by showing that ``off-diagonal'' alignments, i.e. alignments which do not share any edges with $A^*$, are not better than $A^*$ with high probability. While there are many bad alignments which are not entirely off-diagonal, this will be useful as later we can show that a bad alignment $A$ in the range of $\SR$ being better than $A^*$ corresponds to an off-diagonal alignment being better than $A^*$ in a subproblem.

\begin{lemma}\label{lemma:offdiag}
For $\samp{n}$, with probability $1-e^{-\Omega(n)}$, $A > A^*$ for all alignments $A$ such that $A$ and $A^*$ do not share any edges.
\end{lemma}
\begin{proof}
The cost of $A^*$ can be upper bounded using a Chernoff bound: The expected number of substitutions is at most $\rho_s n$, so Fact~\ref{fact:chernoff} gives 

$$Pr\left[A^* \leq \frac{3}{2}\rho_s n\right] \leq 1 - e^{-\frac{\rho_s n}{12}}.$$

Now our goal is to show that with high probability, no alignment $A$ that does not share edges with $A^*$ has cost lower than $cn$ (where $c=\frac{3}{2}\rho_s$). We achieve this using a union bound over alignments, grouping alignments by their number of deletions $d$ (which in the substitution-only case is also the number of insertions). We can ignore alignments with more than $cn/2$ deletions, as they will of course have cost more than $cn$. 

\begin{align*}
  Pr[\exists A, A \leq cn] 
  &\leq \sum_{d=1}^{cn/2}\sum_{ A\text{ with } d \text{ deletions}}Pr[A \leq cn]\\
  &\leq \sum_{d=1}^{cn/2} \binom{n+d}{d,d,n-d} Pr\left[Binom(n-d,\frac{1}{2}) \leq cn - 2d\right]\\
  &\leq \frac{cn}{2}\binom{(1+\frac{c}{2})n}{\frac{c}{2}n,\frac{c}{2}n,(1-\frac{c}{2})n} Pr\left[Binom((1-\frac{c}{2})n,\frac{1}{2})\leq cn\right].
\end{align*}
The second line counts the number of alignments with $d$ deletions, and it expresses the probability of success in terms of the number of substitutions, or edges in $A$ of the form $((i-1,j-1),(i,j))$: The cost of each off-diagonal edge of the form $((i-1,j-1),(i,j))$ is $Bern(\frac{1}{2})$, even if we condition on the cost of all previous edges in $A$: assuming wlog that $i > j$ knowing the costs of all edges before $((i-1,j-1),(i,j))$ in $A$ gives no information about the bit $i$ of $s_1$, which is distributed uniformly at random. So the total cost of these edges is given by $Binom((1-\frac{c}{2})n,\frac{1}{2})$. In the third line we upper bound the probability for simplicity. A Chernoff bound now gives:
\begin{align} 
  Pr\left[Binom((1-\frac{c}{2})n,\frac{1}{2})\leq cn\right] 
  &= Pr\left[Binom((1-\frac{c}{2})n,\frac{1}{2})\leq(1-\frac{2-5c}{2-c}) \frac{1}{2}(1-\frac{c}{2})n\right]\nonumber\\
  &\leq \exp\left(-\frac{(2-5c)^2}{8(2-c)}n\right). \label{eqn:sub-chernoff}
\end{align}

Next we upper bound the trinomial using Stirling's approximation:

\begin{align}
    \binom{(1+\frac{c}{2})n}{\frac{c}{2}n,\frac{c}{2}n,(1-\frac{c}{2})n}
    &\leq   \frac{e}{(2\pi)^{3/2}}
            \frac{((1+\frac{c}{2})n)^{(1+\frac{c}{2})n+\frac{1}{2}}}
            {(\frac{c}{2}n)^{cn+1} ((1-\frac{c}{2})n)^{(1-\frac{c}{2})n +\frac{1}{2}}}
            \nonumber\\
    &\leq   \frac{e}{(2\pi)^{3/2}}
            \frac{2}{cn}\sqrt{\frac{2+c}{2-c}}
            \left[
               \frac{(1+\frac{c}{2})^{(1+\frac{c}{2})}}{(\frac{c}{2})^c (1-\frac{c}{2})^{(1-\frac{c}{2})}} 
            \right]^n.
            \label{eqn:tri-stirling}
\end{align}

Putting everything together, we have the following upper bound

\begin{align*}
    Pr[\exists A, A \leq cn] 
    &\leq \frac{e}{(2\pi)^{3/2}}\frac{2}{cn}\sqrt{\frac{2+c}{2-c}}
    \left[
        \frac{(1+\frac{c}{2})^{(1+\frac{c}{2})}}{(\frac{c}{2})^c (1-\frac{c}{2})^{(1-\frac{c}{2})}} 
    \right]^n 
    \left[
        \exp\left(-\frac{(2-5c)^2}{8(2-c)}\right)
    \right]^n.
\end{align*}

For the above bound to be exponentially decaying in $n$, we need that:
\begin{align}
    \frac{(1+\frac{c}{2})^{(1+\frac{c}{2})}}{(\frac{c}{2})^c (1-\frac{c}{2})^{(1-\frac{c}{2})}}
    \exp\left(-\frac{(2-5c)^2}{8(2-c)}\right)
    &< 1\label{eqn:bound-c},
\end{align}

which holds as long as $c \leq 0.042$, i.e. $\rho_s \leq .028$. For these values of $c$, with high probability $A^*< cn$ and $A > cn$ for any $A$ which does not share any edges with $A^*$.
\end{proof}

We now make the following observations which will allow us to apply Lemma~\ref{lemma:offdiag} to make more powerful statements about the set of all alignments:

\begin{fact}\label{fact:lbr}
Fix any $s_1, s_2, \cond$ in the support of $ID(n)$, and let $A, A'$ be any two alignments with the same set of long breaks. Then $\LR(A) - A = \LR(A') - A'$.
\end{fact}

This follows because applying $\LR$ to $A, A'$ results in the same pairs of subpaths being swapped (and thus the same change in cost) as $A, A'$ have the same long breaks.

\begin{corollary}\label{corollary:sbrrange}
Fix any $(s_1, s_2, \cond)$ in the support of $ID(n)$. If for all alignments $A$ in the range of $\SR$, $A \geq A^*$, then any lowest-cost good alignment is also a lowest-cost alignment.
\end{corollary}
\begin{proof}
Applying a composition of $\LR$ and $\SR$ to any alignment gives $A^*$, and for any $A$, $A$ and $\SR(A)$ have the same long breaks. This gives that for any alignment $A$, $\LR(A)$ (a good alignment) satisfies $\LR(A) \leq A$:

$$\LR(A) - A \stackrel{\textnormal{Fact~\ref{fact:lbr}}}{=} \LR(\SR(A))-\SR(A) = A^* - \SR(A) \leq 0.$$

Now, letting $A'$ be a lowest-cost good alignment, we get $A \geq \LR(A) \geq A'$ for all $A$, i.e. $A'$ is the lowest cost alignment. 
\end{proof}

We complete the argument by showing that the assumption of Corollary~\ref{corollary:sbrrange} holds with high probability.

\begin{lemma}\label{lemma:offdiaggen}
For $\samp{n}$, with probability $1-n^{\Omega(1)}$ for all alignments $A$ in the range of $\SR$, $A \geq A^*$.
\end{lemma}

\begin{proof}
As in Lemma~\ref{lemma:offdiag}, we apply a union bound over the range of $\SR$, grouped by total length of breaks from $A^*$.  Consider the set $\mathcal{A}_i$ contained in the range of $\SR$, which contains all alignments $A$ for which the sum of the lengths of breaks of $A$ from $A^*$ is in $[ik \ln n, (i+1)k\ln n)$.  Then the sets $\{\mathcal{A}_i: 0\leq i\leq \frac{n}{k\ln n} \}$ forms a disjoint cover of the range $\SR(\mathcal{A})$. Note that elements of $\mathcal{A}_i$ have at most $i$ breaks from $A^*$, each of length at least $k\ln n$. Also note that $\mathcal{A}_0$ is a singleton set containing only $A^*$.

For any alignment $A$, we call the set of starting and ending indices of all breaks of that alignment the \textit{breakpoint configuration} of $A$ (to simplify future analysis, we index with respect to $s_1$ \footnote{In the substitution only case, indexing with respect to $s_1$ and $s_2$ is the same, but when indels are present indexing with respect to $s_1$ will simplify the analysis.}). Let $\mathcal{B}_i$ be the set of all possible breakpoint configurations of  alignments in $\mathcal{A}_i$. We can view $B\in\mathcal{B}_i$ as a binary assignment of each edge in $A^*$ to either agree or disagree with $A\in \mathcal{A}_i$. For a fixed set of break points $B\in\mathcal{B}_i$, let $\mathcal{A}_B$ be the set of all alignments having the breakpoints corresponding to $B$ (i.e. every alignment in $\mathcal{A}_B$ has the same breaks from $A^*$). Note that the set $\{\mathcal{A}_B:B\in\mathcal{B}_i\}$ forms a disjoint cover of $\mathcal{A}_i$. 

For any fixed set of breaks $B\in\mathcal{B}_i$, let $s^B_1, s^B_2$ denote the restriction of $s_1, s_2$ to indices contained in the breaks in $B$, and $(A)_B$ denote the restriction of an alignment $A$ to these indices. $s^B_1, s^B_2$ are distributed according to $ID(b)$ for $b \geq ik \ln n$. Furthermore, for $A \in \mathcal{A}_B$, $A < A^*$ if and only if $(A)_B < (A^*)_B$. Since for all $A \in \mathcal{A}_B$, $(A)_B$ does not share any edges with $(A^*)_B$, by Lemma~\ref{lemma:offdiag}:

$$\Pr[\exists A\in \mathcal{A}_B, A< A^*] = \Pr[\exists A\in \mathcal{A}_B, (A)_B< (A^*)_B] \leq e^{-\Omega(ik\ln n)}=n^{-\Omega(ik)}.$$

This reduces our problem to that of counting the cardinality of $\mathcal{B}_i$:
\begin{align*}
    Pr[\exists A\in \SR(\mathcal{A}), A<A^*]
    &= 
    \sum_{i=1}^{\frac{n}{k\ln n}}
            Pr[\exists A\in \mathcal{A}_i, A < A^*]\\
    &= 
    \sum_{i=1}^{\frac{n}{k\ln n}}
        \sum_{B \in \mathcal{B}_i} 
            Pr[\exists A\in \mathcal{A}_B, A < A^*]\\
    &\leq
    \sum_{i=1}^{\frac{n}{k\ln n}}
        \sum_{B \in \mathcal{B}_i} 
            n^{-\Omega(ik)}
    = \sum_{i=1}^{\frac{n}{k\ln n}}
        \left|\mathcal{B}_i\right| 
            n^{-\Omega(ik)}.
\end{align*}

Now we must count the cardinality of $\mathcal{B}_i$. We claim that each $B \in \mathcal{B}_i$ can be uniquely mapped to $i$ or less contiguous subsets of $[n]$, each of a size in $[k \ln n, 2k \ln n)$ or size 0. There are at most $nk \ln n+1$ such subsets (there are $n$ different possible smallest elements for each non-empty subset, and $k \ln n$ different possible sizes for each non-empty subset, and the smallest element and size uniquely determine the non-empty subsets), giving that $$|\mathcal{B}_i| \leq (nk \ln n + 1)^i.$$

Our mapping is as follows: For a break in $B \in \mathcal{B}_i$ which starts at index $j$ and has length $\ell \in [i'k \ln n, (i'+1)k \ln n)$, we map the break to the subsets $\{j, j+1 \ldots j+k\ln n -1\}, \{j+k \ln n, j+k \ln n +1 \ldots j+2k\ln n -1\} \ldots \{j+(i'-1)k \ln n, j+(i'-1)k \ln n+1 \ldots \ell\}$. That is, for a break we take the indices the break spans, and peel off the first $k \ln n$ elements to create a subset, until there are less than $2k \ln n$ indices remaining, which then form their own subset. We map $B$ to the union of the subsets its breaks are mapped to, plus enough empty subsets to make the total number of subsets $i$. It is straightforward to see that this map from $\mathcal{B}_i$ to a set of subsets is injective as desired, and that the set of subsets has the stated properties.

Using $|\mathcal{B}_i| \leq (nk \ln n + 1)^i$ and assuming $k$ is a sufficiently large constant we get:
$$Pr[\exists A\in \SR(\mathcal{A}), A<A^*] \leq \sum_{i=1}^{\frac{n}{k\ln n}} (nk \ln n + 1)^in^{-\Omega(ik)} \leq n^{-\Omega(k)}.$$
\end{proof}

\begin{proof}[Proof of Theorem~\ref{thm:subonly}]
The algorithm is to use the standard DP algorithm restricted to entries indexed by $(i, j)$ where $|i - j| \leq k \ln n$, ignoring dependencies on entries for which $|i - j| > k \ln n$. Theorem~\ref{thm:subonly} follows immediately from Corollary~\ref{corollary:sbrrange}, Lemma~\ref{lemma:offdiaggen}, and the observation that all good alignments are contained in the set of entries used by the DP algorithm.
\end{proof}
\section{Finding an Approximate Alignment}\label{sec:alignment}

We now consider the case where insertions and deletions are present. While in the substitution case it is obvious that the canonical alignment is the diagonal, in the presence of insertions and deletions there is the additional algorithmic challenge of finding something close to the canonical alignment.
We now use our previous definition for alignments to define an alignment function, which will be useful in analyzing the approximate alignment algorithm.

\begin{definition} 
Given an alignment $A$ of $\samp{n}$, let $f_A:[n] \rightarrow \mathbb{Z}$ be the function such that for all $i \in [n]$, $(i, f_A(i))$ is the first vertex in $A$ of the form $(i, j)$.
\end{definition}

Using this definition, $\f(j)$ gives the location of the $j$th bit of $s_1$ in $s_2$, or if the $j$th bit is deleted, where the location would be had it not been deleted. To find the edit distance between $s_1, s_2$, our algorithm will start by computing an approximate alignment function which does not differ much from $\f$. Before describing our algorithm, it will help to prove some properties about edit distances between pairs of strings sampled from $ID(n)$.

\subsection{Properties of the Indel Channel}

The term $(\rho_i\rho_i'+(\rho_d +1/k \ln n)(\rho_d'+1))$, which is roughly speaking an upper bound on the edit distance (divided by $k \ln n)$ between $s_1, s_2$ sampled from $ID(k \ln n)$ due to indels, appears frequently in the rest of the analysis. To simplify the presentation, we denote  $(\rho_i\rho_i'+(\rho_d +1/k \ln n)(\rho_d'+1))$ by $\kappa_n$ for the rest of the paper. Our goal in the following lemmas is to show that by computing the edit distance between the substrings of length $k \ln n$ starting at bit $i_1$ of $s_1$ and bit $i_2$ of $s_2$, we can identify if $i_2 \approx \f(i_1)$. 
\begin{lemma}\label{lemma:smallED}
For $\samp{n}$, let $s_1'$ be the substring formed by bits $i$ to $i + k \ln n-1$ of $s_1$, and $s_2'$ be the substring formed by bits $\f(i)$ to $\f(i+k \ln n)-1$ of $s_2$. Then: 

$$\Pr_{\samp{n}}\left[ED(s_1', s_2') \geq \frac{3}{2}(\rho_s + \kappa_n)k \ln n\right] \leq n^{-\rho_s k/12}+2n^{-\rho_i k /60}+3n^{-\rho_d k / 60}.$$

\end{lemma}
\begin{proof}
The edit distance between $s_1$ and $s_2$ is upper bounded by the number of substitutions, deletions, and insertions that occur in the channel on bits $i$ to $i+k \ln n-1$ of $s_1$. So it suffices to show this total is at most $\frac{3}{2}(\rho_s + \rho_i + \rho_d)k \ln n$ with high probability. In turn, it suffices to show the number of substitutions is at most $\frac{3}{2} \rho_s k \ln n$, the number of insertions is at most $\frac{3}{2} \rho_i \rho_i' k \ln n$, and the number of deletions is at most $\frac{3}{2}(\rho_d k \ln n + 1)(\rho_d' + 1)$ with high probability. We do this using a union bound over the three types of mutations.

The number of substitutions is at most $\rho_s k \ln n$ in expectation. A Chernoff bound with $\epsilon = 1/2$ gives that the number of substitutions exceeds $\frac{3}{2}\rho_s k \ln n$ with probability at most $n^{-\rho_s k/12}$. 
The probability the number of insertions exceeds $\frac{3}{2}\rho_i \rho_i' k \ln n$ is maximized when $p_i = \rho_i, 1/(1-q_i) = \rho_i'$. The number of insertions is then the random variable $NBinom(Binom(k \ln n, \rho_i), 1/\rho_i')$ with expectation $\rho_i \rho_i' k \ln n$, and by Lemma~\ref{lemma:binomofgeo} with $k = 3/2$ the probability it exceeds $\frac{3}{2} \rho_i \rho_i' k \ln n$ is at most $2n^{-\rho_i k /60}$. 

To bound the number of deletions, we consider the following process for deciding where deletions occur in $s_1$:

\begin{itemize}
    \item For each bit of $s_1$ a ``type 1'' deletion occurs with probability $p_d$, except bit 1 of $s_1$ where the probability is $q_d$.
    \item For each bit $j$ where a type 1 deletion occurs, we sample $\delta \sim Geo(\frac{1-q_d}{1-p_d})$. Let $\Delta$ be the number of bits between $j$ and the next bit with a type 1 deletion. A type 2 deletion occurs on the $\min\{\delta, \Delta\}$ bits following $j$. 
\end{itemize}

For bit $1$, its probability of seeing a deletion in the indel channel is upper bounded by $q_d$. Otherwise, if no deletion occurs on bit $j-1$, then for bit $j > i$, the only way bit $j$ sees a deletion is if it has a type 1 deletion, which occurs with probability $p_d$. If a deletion occurs on bit $j-1$ and bit $j$ does not have a type 1 deletion, it sees a type 2 deletion with probability $(1-\frac{1-q_d}{1-p_d}) = \frac{q_d-p_d}{1-p_d}$ by the properties of the geometric distribution (this is regardless of the type of deletion on bit $j-1$). So its overall probability of seeing a deletion is $p_d + (1-p_d)\frac{q_d - p_d}{1-p_d} = q_d$.  So, the number of deletions in this process stochastically dominates the number of deletions on bits $i$ to $i + k \ln n - 1$ of $s_1$. 

Then, the number of deletions is stochastically dominated by the random variable $X+Y$ arrived at by sampling $Y \sim Binom(k \ln n-1, p_d)+Bern(q_d), X \sim NBinom(Y, \frac{1-q_d}{1-p_d})$, which exceeds $\frac{3}{2}(\rho_d + \rho_d \rho_d')k \ln n$ with maximum probability when $p_d = \rho_d$, $\frac{1-\rho_d}{1-q_d} = \rho_d'$. The probability $Y$ exceeds $\frac{3}{2}(\rho_d k \ln n)+1$ is at most $n^{-\rho_d k /12}$ by a Chernoff bound. The probability $X$ exceeds $\frac{3}{2}(\rho_d k \ln n+1)\rho_d'$ is at most $2n^{-\rho_d k /60}$ by Lemma~\ref{lemma:binomofgeo} with $k = 3/2$. So by a union bound the probability the number of deletions exceeds $\frac{3}{2}(\rho_d k \ln n + 1)(\rho_d' + 1)$ is at most $3n^{-\rho_d k / 60}$.
\end{proof}

\begin{lemma}\label{lemma:largeED}
Let $s_1, s_2$ be bitstrings of length $k \ln n$, chosen independently and uniformly at random from all bitstrings of length $k \ln n$. Then $Pr[ED(s_1, s_2) \leq D] \leq \frac{(4e\frac{k \ln n}{D} + 5e+\frac{4e}{D})^D}{2^{k \ln n}}.$
\end{lemma}

The proof of this lemma is fairly standard (see e.g. \cite[Lemma 8]{Batu2003ASA}). For completeness, we provide a proof here.

\begin{proof}
We first bound the number of strings within edit distance $D$ of $s_1$. Fix any set of up to $D$ edits that can be applied to a bitstring initially of length $k \ln n$, that does not contain redundant edits (such as substituting the same bit more than once, deleting an inserted bit). This set can be mapped to a set of $D$ tuples as follows:
\begin{itemize}
    \item For a substitution (or deletion) applied to the bit in the $i$th position (using the indexing prior to insertions and deletions), it is encoded as the tuple $(i, S)$ (or $(i, D)$ for a deletion). Note that by the assumption that there are no redundant edits, all substitution and deletion edits in the set of edits map to distinct tuples.
    \item For insertions, we handle indexing differently to still ensure no two insertions are mapped to the same tuple. Suppose the set of $D$ edits inserts the bitstring $b_1b_2\ldots b_k$ to the right of index of $i$ (using the original indexing - we treat bits are being inserted to the left of the entire bitstring as being inserted to the right of bit 0). Let $i'$ be $i$ plus the number of insertions in the set of edits occurring before bit $i$. Then we map these $k$ insertions to the tuples $(i', I_{b_1}), (i'+1, I_{b_2}) \ldots (i'+k-1, I_{b_k})$. This ensures that the insertions in the set of edits also get mapped to different tuples, since the first index will be distinct for all tuples that insertions are mapped to. 
    \item If the number of edits is $D - k$, we include $(1, N), (2, N) \ldots (k, N)$ in the final set of tuples so the final set of tuples still has size $D$.
\end{itemize}

Every tuple that can be mapped to in this encoding scheme is of the form $(i, E)$ for $0 \leq i \leq k \ln n + D, E \in \{S, D, I_0, I_1\}$ or $(i, N)$ for $1 \leq i \leq D$. So, there are at most $\binom{4k \ln n+5D+4}{D}$ sets of $D$ tuples that any set of up to $D$ edits can be mapped to. Furthermore, note that the mapping is injective, i.e. given a set of $D$ tuples, using the reverse of the above process it can be uniquely mapped to set of edits. So, there are also at most $\binom{4k \ln n+5D+4}{D}$ possible ways to apply at most $D$ edits to a bitstring which is initially length $k \ln n$. Stirling's approximation gives that this is at most $(4e\frac{k \ln n}{D} + 5e+\frac{4e}{D})^D$. 
So there are at most $(4e\frac{k \ln n}{D} + 5e+\frac{4e}{D})^D$ strings $s'$ such that $ED(s_1, s') \leq D$. The number of bitstrings of length $k \ln n$ is $2^{k \ln n}$. So the probability $ED(s_1, s_2) \leq D$ is at most $\frac{(4e\frac{k \ln n}{D} + 5e+\frac{4e}{D})^D}{2^{k \ln n}}$.
\end{proof}

\begin{lemma}\label{lemma:localshifts}
For constant $k > 0$, $i \leq n - k \ln n$, 

$$Pr_{\samp{n}}\left[ |\f(i + k \ln n) - \f(i) - k \ln n| \leq \frac{3}{2}\kappa_n\cdot k \ln n \right] \geq$$
$$1-2n^{-\rho_i k /60}-3n^{-\rho_d k / 60}.$$
\end{lemma}
\begin{proof}
$\f(i + k \ln n) - \f(i) - k \ln n$ is the signed difference between the number of insertions and deletions happening in indices $i$ to $i + k \ln n - 1$ of $s_1$. A simple upper bound for this difference is the sum of the number of insertions and deletions. The same analysis as Lemma~\ref{lemma:smallED} gives the lemma.
\end{proof}

\begin{corollary}\label{cor:smallED}
Consider the following random process, which we denote $\mathcal{P}$: we choose $i_1$ such that $i_1 < n - k \ln n$, sample $\samp{n}$, and then choose an arbitrary $i_2$ such that $|i_2 - \f(i_1)| \leq \ln n$ and $i_2$ is at least $k \ln n$ less than the length of $s_2$. Let $s_1'$ denote the string consisting of bits $i_1$ to $i_1 + k \ln n - 1$ of $s_1$ and $s_2'$ the string consisting of bits $i_2$ to $i_2 + k \ln n - 1$ of $s_2$. Then for any $i_2$ we choose satisfying the above conditions,
$$\Pr_{\mathcal{P}}\left[ED(s_1', s_2') \leq (1 + \frac{3}{2}(\rho_s + 2\kappa_n)) k\ln n\right] \geq$$
$$1 - 2n^{-\rho_i k /12}- 4n^{-\rho_i k /60} - 6n^{-\rho_d k /60}.$$
\end{corollary}

\begin{proof}
By Lemma~\ref{lemma:localshifts} and the assumptions in the corollary statement, with probability at least $1-2n^{-\rho_i k /60}-3n^{-\rho_d k / 60}$, the edit distance between $s_2'$ and bits $\f(i_1)$ to $\f(i_1 + k \ln n)-1$ of $s_2$ (call this substring $s_2^*$) is at most $(1 + \frac{3}{2}\kappa_n)k \ln n$ (the upper bound on the difference between starting indices plus the high-probability upper bound on the difference between ending indices). $s_2^*$ is the result of passing $s_1'$ through the indel channel, so by Lemma~\ref{lemma:smallED} with probability at least $1 - n^{-\rho_s k/12}-2n^{-\rho_i k /60}-3n^{-\rho_d k / 60}$, the edit distance between $s_2^*$ and $s_1'$ is at most $\frac{3}{2}(\rho_s + \kappa_n)k \ln n$, giving the lemma by a union bound and triangle inequality.
\end{proof}

\begin{corollary}\label{cor:largeED}
Consider the following random process, which we denote $\mathcal{P}$: we choose $i_1$ such that $i_1 < n - k \ln n$, sample $\samp{n}$, and then choose an arbitrary $i_2$ such that 

$$|i_2 - \f(i_1)| > \left(\frac{3}{2}\kappa_n + 1\right) k \ln n,$$

and $i_2$ is at least $k \ln n$ less than the length of $s_2$. Let $s_1'$ denote the string consisting of bits $i_1$ to $i_1 + k \ln n - 1$ of $s_1$ and $s_2'$ the string consisting of bits $i_2$ to $i_2 + k \ln n - 1$ of $s_2$. Then for $0 < r < 1$, 

$\Pr_{\mathcal{P}}[ED(s_1', s_2') > kr \ln n] \geq 1 - \left[\frac{(\frac{4e}{r} + 5e + \frac{4e}{kr \ln n})^r}{2} \right]^{k \ln n}-2n^{-\rho_i k /60}-3n^{-\rho_d k/ 60}.$
\end{corollary}
\begin{proof}
Either $i_2 < \f(i_1) - k \ln n$ or $i_2 > \f(i_1) + (\frac{3}{2}\kappa_n+1)k \ln n$. If $i_2 < \f(i_1) - k \ln n$, then none of the bits in $s_2'$ are inherited from bits in $s_1'$. If $i_2 > \f(i_1) + (\frac{3}{2}\kappa_n+1)k \ln n$, then by Lemma~\ref{lemma:localshifts} we have with probability $1-2n^{-\rho_i k /60}-3n^{-\rho_d k / 60}$:

$$i_2 - \f(i_1 + k \ln n) = [i_2 - \f(i_1) - k \ln n] + [\f(i_1) + k \ln n - \f(i_1 + k \ln n)] \geq$$
$$\frac{3}{2}\kappa_n \cdot k \ln n - \frac{3}{2}\kappa_n \cdot k \ln n = 0.$$

Then since $i_2 > \f(i_1 + k \ln n)$, none of the bits are in $s_2'$ are inherited from bits in $s_1'$. In either case, $s_1', s_2'$ are independent and uniformly random bitstrings, and we can apply Lemma~\ref{lemma:largeED} with $D = kr \ln n$ to get the lemma by a union bound.
\end{proof}

Let $n_0$ be a sufficiently large constant. If we choose any $r$ which is less than a certain constant (which is approximately .1569), for all $n \geq n_0$, if $k$ is sufficiently large then the term $\frac{(\frac{4e}{r} + 5e + \frac{4e}{kr \ln n})^r}{2}$ from Corollary~\ref{cor:largeED} is less than 1 and thus the failure probability in Corollary~\ref{cor:largeED} becomes $n^{-\Omega(k)}$. If for all $n \geq n_0$, $(1 + \frac{3}{2}k(\rho_s + 2\kappa_n)) < kr$, then for all $n \geq n_0$ the lower bound on edit distance given by Corollary~\ref{cor:largeED} exceeds the upper bound given by Corollary~\ref{cor:smallED}. In turn, informally we have the desired property that we can use the edit distance between substrings of length $k \ln n$ in $s_1$ and $s_2$ to test if these substrings are close in the canonical alignment. So for the rest of this section, we will fix $\rho_s, \rho_i, \rho_i', \rho_d, \rho_d', r$ to be positive values satisfying these conditions for all $n \geq n_0$. Once these values are fixed we can make the failure probabilities in both corollaries $n^{-c}$ with any exponent $c$ of our choice ($c = 2$ will suffice to achieve a final failure probability of $O(1/n)$) by choosing a sufficiently large $k$ depending only on $c$ and $n_0$. So we also fix $k$ to be said sufficiently large value.

\subsection{Algorithm for Quickly Finding an Approximate Alignment}

We now describe the algorithm \textsc{ApproxAlign}, given as Algorithm 1, which finds the approximate alignment $f'$. Informally, \textsc{ApproxAlign} runs as follows: It starts by initializing $f'(1) = 1$, which is of course exactly correct. By Lemma~\ref{lemma:localshifts}, we know that  $\f(k \ln n + 1)$ will be within $O(\ln n)$ of $1 + k \ln n$. So, to decide what $f'(k \ln n + 1)$ will be, we compute the edit distance between bits $k \ln n + 1$ to $2k \ln n$ of $s_1$ and bits $j$ to $j + k \ln n - 1$ of $s_2$ for various values of $j$ close to $1 + k \ln n$. By Corollary~\ref{cor:smallED} we know that when $j$ is near  $\f(k \ln n + 1)$, the edit distance will be small, and by Corollary~\ref{cor:largeED} we know that when $j$ is far from  $\f(k \ln n + 1)$ the edit distance will be large. So whichever value of $j$ causes the edit distance to be minimized is not too far from the true value of  $\f(k \ln n + 1)$. Once we've decided on the value $f'(k \ln n + 1)$, we proceed analogously to choose a value for $f'(2k \ln n + 1)$, using $f'(k \ln n+1)$ to decide what range of values try, and so on. We now formally prove our guarantee for \textsc{ApproxAlign} (including the runtime guarantee).

\begin{algorithm}\label{alg:approxalign}
  \caption{Algorithm for Approximate Alignment}
  \begin{algorithmic}[1]
    \Function{\textsc{ApproxAlign}}{$s_1$, $s_2$}
      \State $f'(1) \leftarrow 1$
      \State $J  \leftarrow 2 \lceil(\frac{3}{2}\kappa_n+1) \cdot k\rceil$
      \For{$i = 1, 2, \ldots \lfloor \frac{n}{k \ln n} \rfloor-1 $}
      \State $minED \leftarrow \infty$
      \For{$j = -J, -J+1, \ldots J$} \par 
      \State $s_1' \leftarrow$ bits $ik \ln n + 1$ to $(i+1)k \ln n$ of $s_1$ 
      \State $s_2' \leftarrow$ bits $f'((i-1)k \ln n+1)+(j+k)\ln n$ to \par\qquad\qquad $f'((i-1)k \ln n+1)+(j+2k)\ln n-1$ of $s_2$
      \If{$ED(s_1', s_2') \leq minED$}
      \State $minED \leftarrow ED(s_1', s_2')$
      \State $f'(ik\ln n+1) \leftarrow f'((i-1)k \ln n)+(j+k) \ln n$
      \EndIf 
      \EndFor
      \EndFor
      \State \Return $f'$
    \EndFunction
  \end{algorithmic}
\end{algorithm}

\begin{lemma}\label{lemma:approxalign}
For $\samp{n}$,  \textsc{ApproxAlign}$(s_1, s_2)$ computes in time $O(n \ln n)$ a function $f'$ such that with probability at least $1-n^{-\Omega(1)}$, for all $i$ where $f'(i)$ is defined $|f'(i) - \f(i)| \leq \lceil (\frac{3}{2}\kappa_n+1) k \ln n \rceil$.
\end{lemma}
\begin{proof}
We proceed by induction. Clearly, $|f'(1) - \f(1)| = |1 - 1| \leq \lceil (\frac{3}{2}\kappa_n+1) \cdot k \ln n \rceil$. Suppose $|f'((i-1) k \ln n + 1) - \f((i-1) k \ln n + 1)| \leq \lceil (\frac{3}{2}\kappa_n+1) \cdot k \ln n \rceil$. By Lemma~\ref{lemma:localshifts} and our choices of constants, with probability $1 - n^{-\Omega(1)}$, $|\f((i-1)k \ln n + 1) + k \ln n - \f(ik \ln n + 1)| \leq \frac{3}{2} \kappa_n \cdot k \ln n$. This gives:

\begin{align*}
&\left|[f'((i-1) k \ln n + 1) + k \ln n] - \f(ik \ln n + 1)\right| \leq \\
&\left|[f'((i-1) k \ln n + 1) + k \ln n] - [\f((i-1) k \ln n + 1) + k \ln n]\right|\\
&\qquad  +\left|[\f((i-1)k \ln n + 1) + k \ln n] - \f(ik \ln n + 1)\right| =\\
&\left|f'((i-1) k \ln n + 1) - \f((i-1) k \ln n + 1)\right| \\
&\qquad + \left|\f((i-1)k \ln n + 1) + k \ln n - \f(ik \ln n + 1)\right| \leq \\
&\left\lceil (\frac{3}{2}\kappa_n+1) \cdot k \ln n \right\rceil + \frac{3}{2}\kappa_n \cdot k \ln n \leq J.\\
\end{align*}
So for some $j$ in the range iterated over by the algorithm, $|f'((i-1)k \ln n+1)+(j+k)\ln n - \f(ik \ln n + 1)| \leq \ln n$ and thus the minimum edit distance $minED$ found by the algorithm in iterating over the $j$ values is at most $(1 + \frac{3}{2}k(\rho_s +2 \kappa_n) \ln n < kr \ln n$ by Corollary~\ref{cor:smallED} with probability at least $1 - n^{-\Omega(1)}$. By Corollary~\ref{cor:largeED}, with probability at least $1 - n^{-\Omega(1)}$ the final value of $f'(ik \ln n + 1)$ can't differ from  $\f(ik \ln n + 1)$ by more than $\lceil (\frac{3}{2}\kappa_n+1) \cdot k \ln n \rceil$ as desired - otherwise, by the corollary with high probability $minED$ would be larger than $kr \ln n$.

Thus by induction, $|f'(i) - \f(i)| \leq \lceil (\frac{3}{2}\kappa_n+1) \cdot k \ln n \rceil$ for all $i$ if the high probability events of Lemma~\ref{lemma:localshifts}, Corollary~\ref{cor:smallED}, and Corollary~\ref{cor:largeED} occur in all inductive steps. Across all inductive steps we require $O(n)$ such events to occur, and each occurs with probability $1-n^{-\Omega(1)}$ where the negative exponent can be made arbitrarily large, so by a union bound we can conclude that with probability $1-n^{-\Omega(1)}$, $|f'(i) - \f(i)| \leq 2k \ln n$ for all $i$.

For runtime, note that the for loops iterate over $O(\frac{n}{\ln n})$ values of $i$ and $O(1)$ values of $j$. For each $i, j$ pair, we perform an edit distance computation between two strings of length $O(\ln n)$ which can be in done in $O(\ln^2 n)$ time using the canonical dynamic programming algorithm. So the overall runtime is $O(n \ln n)$. 
\end{proof}
\section{Error Analysis with Indels}\label{sec:indels}

In this section, we extend the results from Section~\ref{sec:subonly} to the case where indels are present. 

\begin{lemma}\label{lemma:offdiag-id}
For any realization of $\samp{n}$, let $s_1'$ be the restriction of $s_1$ to any fixed subset of indices $B$ of total size $\ell \geq k \ln n$, $s_2'$ be the substring of $s_2$ that $A^*$ aligns with $s_1'$, and let $(A^*)_B$ denote the restriction of the alignment $A^*$ to indices in $s_1', s_2'$. Then with probability $1-e^{-\Omega(\ell)}$ over $\samp{n}$, $A > (A^*)_B$ for all alignments $A$ of $s_1', s_2'$ such that $A$ and $(A^*)_B$ do not share any edges.
\end{lemma}

\begin{proof}

We proceed similarly to Lemma~\ref{lemma:offdiag}, but for the case with indels. The same analysis as in Lemma~\ref{lemma:smallED} gives that that for a fixed $s_1'$, $$\Pr_{\samp{n}}[(A^*)_B\geq \frac{3}{2}\left(\rho_s + \kappa_n \right)\ell] \leq e^{-\rho_s \ell /12}+2e^{-\rho_i \ell /60}+3e^{-\rho_d \ell / 60}.$$
Our goal now is to show any alignment $A$ of $s_1', s_2'$ that shares no edges with $(A^*)_B$ has $A>c\ell$ with high probability, where $c = \frac{3}{2}\rho_s + \kappa_n$.

Fix any realization $\zeta$ of the positions of indels generated by $\samp{n}$, without fixing the values of $s_1$, the inserted bits, or the positions of substitutions. Let $\ell_1=\ell$ and $\ell_2$ be the lengths of $s_1'$ and $s_2'$. Let $r=|\ell_1 - \ell_2|$. A similar analysis to Lemma~\ref{lemma:localshifts} gives that $r \leq \kappa_n \ell$ with probability $1-e^{-\Omega(\ell)}$, so it suffices to prove the lemma statement holds with high probability conditioned on any $\zeta$ such that $r < \kappa_n \ell$, so we condition on $\zeta$ for the rest of the proof. Assume without loss of generality that $\ell_2 - \ell_1 = r$, i.e. that the $r$ excess indels are insertions. The counting argument is similar when $\ell_1 - \ell_2 = r$. As before, we sum over the number of deletions, $d$, which corresponds to $d+r$ insertions and $\ell-d$ substitutions.
\begin{align*}
     Pr[\exists A, A \leq c\ell]
     &\leq \sum_{d=0}^{c\ell/2}\sum_{A\in\mathcal{A} \text{ with } d \text{ deletions}} Pr[A\leq c\ell]\\
     &\leq \sum_{d=0}^{c\ell/2}\binom{\ell+d+r}{d,d+r,\ell-d} Pr[Binom(\ell-d,\frac{1}{2})\leq c\ell - 2d-r]\\
     &\leq \frac{c\ell}{2}
        \binom{(1+\frac{c}{2})\ell + r}
            {\frac{c}{2}\ell,
             \frac{c}{2}\ell + r,
             (1-\frac{c}{2})\ell} 
        Pr[Binom((1-\frac{c}{2})\ell,\frac{1}{2})\leq c\ell].
\end{align*}
  Where the probability is taken over the events we haven't conditioned on, i.e. the realization of $s_1$, the inserted bits, and the positions of substitutions. Since we assume $r<\kappa_n \ell$, then 
  $\binom{(1+\frac{c}{2})\ell + r}
    {\frac{c}{2}\ell,\frac{c}{2}\ell + r,(1-\frac{c}{2})\ell} 
   \leq 
   \binom{(1+\frac{c}{2} + \kappa_n)\ell} 
     {\frac{c}{2}\ell,(\frac{c}{2}+\kappa_n)\ell,(1-\frac{c}{2})\ell}$ 
  with high probability. Note also that $\kappa_n \leq \frac{3}{2}\rho_s + \kappa_n < c$. Hence, similar to Equation~(\ref{eqn:tri-stirling}) from Lemma~\ref{lemma:offdiag}, Stirling's approximation gives an upperbound on the trinomial
  $$
    \binom{(1+\frac{c}{2} + \kappa_n)\ell} 
        {\frac{c}{2}\ell,
        (\frac{c}{2}+\kappa_n) \ell,
        (1-\frac{c}{2})\ell}
  \leq   
    \frac{e}{(2\pi)^{3/2}}
    \frac{2}{c\ell}
    \sqrt{\frac{2+3c}{2-c}}
    \left[
       \frac{(1+\frac{3}{2}c)^{(1+\frac{3}{2}c)}}
       { 
            (\frac{c}{2})^{c}
            (1-\frac{c}{2})^{(1-\frac{c}{2})}
       } 
    \right]^\ell.
  $$
We combine this with Equation~(\ref{eqn:sub-chernoff}) from Lemma~\ref{lemma:offdiag} for the term $Pr[Binom((1-\frac{c}{2})\ell,\frac{1}{2})\leq c\ell]$, to get that when $c < 0.03485$, the probability decays exponentially in $\ell$. Hence requiring that $\frac{3}{2}\rho_s + \kappa_n < .03485$ ensures that $A>(A^*)_B$ with high probability.
\end{proof}

\begin{lemma}\label{lemma:offdiaggen-id}
For $\samp{n}$, with probability $1-n^{-\Omega(1)}$ for all alignments $A$ in the range of $\SR$, $A \geq A^*$.
\end{lemma}

\begin{proof}

The proof proceeds similarly to that of Lemma~\ref{lemma:offdiaggen}. Recall that the starting/ending indices and the lengths of breaks are defined with respect to the indices in $s_1$. Since $s_1$'s length is $n$ always, we can define sets of break points independently of the realization of $ID(n)$, and so we define $\mathcal{A}_i$, $\mathcal{B}_i$, $\mathcal{A}_B$ as in Lemma~\ref{lemma:offdiaggen}. The restriction of $s_1$ to a fixed subset of indices in the statement~\ref{lemma:offdiag-id} can be applied to the subsets of indices contained in breaks, so the same analysis as in Lemma~\ref{lemma:offdiaggen} gives: 
$$Pr[\exists A\in \SR(\mathcal{A}), A<A^*]
    = \sum_{i=1}^{\frac{n}{k\ln n}}
        \left|\mathcal{B}_i\right| 
            n^{-\Omega(ik)}.$$
Since breakpoints are defined with respect to the fixed-length string $s_1$, as before we have $|\mathcal{B}_i| \leq (nk \ln n + 1)^i$ and thus $Pr[\exists A\in \SR(\mathcal{A}), A<A^*] \leq n^{-\Omega(k)}$ as desired.
\end{proof}

\begin{proof}[Proof of Theorem~\ref{thm:main}]
We estimate $f_{A^*}$ using \textsc{ApproxAlign} to obtain $f'$. Then, we use the standard DP algorithm restricted to entries that are within distance $k_2\ln n$ from $(i,f'(i))$ for some $i$. By Theorem~\ref{lemma:approxalign}, for any fixed $k$, if $k_2$ is sufficiently large, this range of entries computed contains all entries within distance $k \ln n$ of $A^*$, i.e. contains the range of $\LR$. Fact~\ref{fact:lbr} and Corollary~\ref{corollary:sbrrange} also hold when indels are present, so by Lemma~\ref{lemma:offdiaggen-id}, the optimality of the DP algorithm gives that the algorithm is correct.

For runtime, note that \textsc{ApproxAlign} runs in $O(n \ln n)$ time per Theorem~\ref{lemma:approxalign} and the set of entries considered by the DP algorithm is size at most $O(n \ln n)$ (each of the $n / \ln n$ indices where $f'$ is defined contribute $O(\ln^2 n)$ entries to be computed), and each entry can be computed in constant time. So the overall runtime is $O(n \ln n)$ as desired.
\end{proof}

\section*{Acknowledgements}

We thank Satish Rao for suggesting the problem and for pointing out the connection to alignment heuristics used in practice. We thank Nir Yosef for helpful discussions on models for indels used in computational biology. We thank the anonymous reviews for their helpful feedback regarding the presentation of the results.

\bibliographystyle{alpha}
\bibliography{ref}

\end{document}